\documentclass[twoside,openbib,onecolumn]{article}%
\usepackage{amssymb}
\usepackage{amsfonts}
\usepackage{amsmath}
\usepackage{graphicx}%
\providecommand{\U}[1]{\protect\rule{.1in}{.1in}}
\newtheorem{theorem}{Theorem}

\newtheorem{lemma}[theorem]{Lemma}

\newenvironment{proof}[1][Proof]{\noindent\textbf{#1.} }{\ \rule{0.5em}{0.5em}}
\ifx\pdfoutput\relax\let\pdfoutput=\undefined\fi
\newcount\msipdfoutput
\ifx\pdfoutput\undefined\else
\ifcase\pdfoutput\else
\ifx\paperwidth\undefined\else
\ifdim\paperheight=0pt\relax\else\pdfpageheight\paperheight\fi
\ifdim\paperwidth=0pt\relax\else\pdfpagewidth\paperwidth\fi
\fi\fi\fi
\begin{document}

\title{\textbf{Periodic Solutions of 2D Isothermal Euler-Poisson Equations with
Possible Applications to Spiral and Disk-like Galaxies}}
\author{M\textsc{an Kam Kwong\thanks{E-mail address: mankwong@polyu.edu.hk }}\\\textit{Department of Applied Mathematics,}\\\textit{The Hong Kong Polytechnic University,}\\\textit{Hung Hom, Kowloon, Hong Kong}
\and M\textsc{anwai Yuen\thanks{Corresponding author and E-mail address:
nevetsyuen@hotmail.com }}\\\textit{Department of Mathematics and Information Technology,}\\\textit{The Hong Kong Institute of Education,}\\\textit{10 Lo Ping Road, Tai Po, New Territories, Hong Kong}}
\date{Revised 29-May-2014}
\maketitle

\begin{abstract}
Compressible Euler-Poisson equations are the standard self-gravitating models
for stellar dynamics in classical astrophysics. In this article, we construct
periodic solutions to the isothermal ($\gamma=1$) Euler-Poisson equations in
$R^{2}$ with possible applications to the formation of plate, spiral galaxies
and the evolution of gas-rich, disk-like galaxies. The results complement
Yuen's solutions without rotation (M.W. Yuen, \textit{Analytical Blowup
Solutions to the 2-dimensional Isothermal Euler-Poisson Equations of Gaseous
Stars}, J. Math. Anal. Appl. \textbf{341 (}2008\textbf{), }445--456.). Here,
the periodic rotation prevents the blowup phenomena that occur in solutions
without rotation. Based on our results, the corresponding $3$D rotational
results for Goldreich and Weber's solutions are conjectured.

\ 



\end{abstract}

\section{Introduction}

The evolution of self-gravitating galaxies or gaseous stars in astrophysics
can be described by the compressible Euler-Poisson equations:
\begin{equation}
\left\{
\begin{array}
[c]{rl}%
{\normalsize \rho}_{t}{\normalsize +\nabla\cdot(\rho\vec{u})} &
{\normalsize =}{\normalsize 0}\\
\rho(\vec{u}_{t}+(\vec{u}\cdot\nabla)\vec{u}){\normalsize +\nabla P} &
{\normalsize =}{\normalsize -\rho\nabla\Phi}\\
{\normalsize \Delta\Phi(t,\vec{x})} & {\normalsize =\alpha(N)}%
{\normalsize \rho}\text{,}%
\end{array}
\right.  \label{Euler-Poisson2DRotation}%
\end{equation}
where $\alpha(N)$ is a constant related to the unit ball in $R^{N}$, such that
$\alpha(1)=2$, $\alpha(2)=2\pi$, and for $N\geq3$%
\begin{equation}
\alpha(N)=N(N-2)V(N)=N(N-2)\frac{\pi^{N/2}}{\Gamma(N/2+1)}\text{,}%
\end{equation}
where $V(N)$ is the volume of the unit ball in $R^{N}$ and $\Gamma$ is the
Gamma function. The unknown functions $\rho=\rho(t,\vec{x})$ and $\vec{u}%
=\vec{u}(t,\vec{x})=(u_{1},u_{2},....,u_{N})\in\mathbf{R}^{N}$ are the density
and the velocity, respectively. The $\gamma$-law is usually imposed on the
pressure term:%
\begin{equation}
P={\normalsize P}\left(  \rho\right)  {\normalsize =K\rho}^{\gamma}%
\end{equation}
with the constant $\gamma\geq1$. In addition, the ideal fluid is called
isothermal if $\gamma=1$. The Poisson equation (\ref{Euler-Poisson2DRotation}%
)$_{3}$ can be solved as%
\begin{equation}
{\normalsize \Phi(t,\vec{x})=}\int_{R^{N}}Green(\vec{x}-\vec{y})\rho(t,\vec
{y}){\normalsize d\vec{y}}\text{,}%
\end{equation}
with the Green's function%
\begin{equation}
Green(\vec{x})=\left\{
\begin{array}
[c]{ll}%
\log|\vec{x}| & \text{for }N=2\\
\frac{-1}{|\vec{x}|^{N-2}} & \text{for }N\geq3.
\end{array}
\right.
\end{equation}
For $N=3$, the Euler-Poisson equations (\ref{Euler-Poisson2DRotation}) are the
classical models in stellar dynamics given in \cite{BT}, \cite{C}, \cite{KW}
and \cite{Longair}. Some results on local existence of the system can be found
in \cite{M}, \cite{B}, and \cite{G}.

If we seek solutions with radial symmetry, the Poisson equation
(\ref{Euler-Poisson2DRotation})$_{3}$ is transformed to%
\begin{equation}
{\normalsize r^{N-1}\Phi}_{rr}\left(  {\normalsize t,r}\right)  +\left(
N-1\right)  r^{N-2}\Phi_{r}\left(  {\normalsize t,r}\right)  {\normalsize =}%
\alpha\left(  N\right)  {\normalsize \rho r^{N-1}}%
\end{equation}%
\begin{equation}
\Phi_{r}=\frac{\alpha\left(  N\right)  }{r^{N-1}}\int_{0}^{r}\rho
(t,s)s^{N-1}ds.
\end{equation}

In particular, radially symmetric solutions without rotation can be expressed
as%
\begin{equation}
\rho(t,\vec{x})=\rho(t,r)\text{, }\vec{u}(t,\vec{x})=\frac{\vec{x}}{r}V(t,r)
\end{equation}
with the radius $r:=\left(  \sum_{i=1}^{N}x_{i}^{2}\right)  ^{1/2}$. In 1980,
Goldreich and Weber first constructed analytical blowup (collapsing) solutions
of the $3$D Euler-Poisson equations for $\gamma=4/3$ for the non-rotating gas
spheres \cite{GW}. In 1992, Makino \cite{M1} provided a rigorous proof of the
existence of these kinds of blowup solutions. In 2003, Deng, Xiang and Yang
\cite{DXY} generalized the solutions to higher dimensions, $R^{N}$($N\geq3$).
In 2008, Yuen constructed the corresponding solutions (without compact
support) in $R^{2}$ with $\gamma=1$ \cite{YuenJMAA2008a}. In summary, the
family of the analytical solutions is as follows:\newline for $N\geq3$ and
$\gamma=(2N-2)/N$, in \cite{DXY}
\begin{equation}
\left\{
\begin{array}
[c]{c}%
\rho(t,r)=\left\{
\begin{array}
[c]{c}%
\dfrac{1}{a(t)^{N}}f(\frac{r}{a(t)})^{N/(N-2)}\text{ for }r<a(t)S_{\mu}\\
0\text{ for }a(t)S_{\mu}\leq r
\end{array}
\right.  \text{, }V{\normalsize (t,r)=}\dfrac{\dot{a}(t)}{a(t)}{\normalsize r}%
\\
\ddot{a}(t){\normalsize =}\dfrac{-\lambda}{a(t)^{N-1}},\text{ }%
{\normalsize a(0)=a}_{0}>0{\normalsize ,}\text{ }\dot{a}(0){\normalsize =a}%
_{1}\\
\ddot{f}(s){\normalsize +}\dfrac{N-1}{s}\dot{f}(s){\normalsize +}\dfrac
{\alpha(N)}{(2N-2)K}f{\normalsize (s)}^{N/(N-2)}{\normalsize =}\frac
{{\normalsize N(N-2)\lambda}}{{\normalsize (2N-2)K}}{\normalsize ,}\text{
}f(0)=\alpha>0,\text{ }\dot{f}(0)=0\text{,}%
\end{array}
\right.  \label{3-dgamma=4over3}%
\end{equation}
where the finite $S_{\mu}$ is the first zero of $f(s)$ and\newline for $N=2$
and $\gamma=1$, in \cite{YuenJMAA2008a}%
\begin{equation}
\left\{
\begin{array}
[c]{c}%
\rho(t,r)=\dfrac{1}{a(t)^{2}}e^{f(\frac{r}{a(t)})}\text{, }%
V{\normalsize (t,r)=}\dfrac{\dot{a}(t)}{a(t)}{\normalsize r}\\
\ddot{a}(t){\normalsize =}\dfrac{-\lambda}{a(t)},\text{ }{\normalsize a(0)=a}%
_{0}>0{\normalsize ,}\text{ }\dot{a}(0){\normalsize =a}_{1}\\
\ddot{f}(s){\normalsize +}\dfrac{1}{s}\dot{f}(s){\normalsize +\dfrac{2\pi}%
{K}e}^{f(s)}{\normalsize =\frac{2\lambda}{K},}\text{ }f(0)=\alpha,\text{ }%
\dot{f}(0)=0.
\end{array}
\right.  \label{solution 3 Yuen}%
\end{equation}

Similar solutions exist for other similar systems, see, for example,
\cite{YuenNonlinearity2009} and \cite{YuenCQG2009}. All the above known
solutions are without rotation.

For the 2D Euler equations with $\gamma=2$,
\begin{equation}
\left\{
\begin{array}
[c]{rl}%
{\normalsize \rho}_{t}{\normalsize +\nabla\cdot(\rho\vec{u})} &
{\normalsize =}{\normalsize 0}\\
\rho(\vec{u}_{t}+(\vec{u}\cdot\nabla)\vec{u}){\normalsize +\nabla P} &
{\normalsize =0,}%
\end{array}
\right.
\end{equation}
Zhang and Zheng \cite{ZZ} in 1995 constructed the following explicitly spiral
solutions:%
\begin{equation}
\rho=\frac{r^{2}}{8Kt^{2}}\text{, }u_{1}=\frac{1}{2t}(x+y)\text{, }u_{2}%
=\frac{1}{2t}(x-y)
\end{equation}
in $r\leq2t\sqrt{\dot{P}_{0}}$, and%
\begin{equation}
\left\{
\begin{array}
[c]{c}%
\rho=\rho_{0},\\
u_{1}=(2t\dot{P}_{0}\cos\theta+\sqrt{2\dot{P}_{0}}\sqrt{r^{2}-2t^{2}\dot
{P}_{0}}\sin\theta)/r,\\
u_{2}=(2t\dot{P}_{0}\sin\theta-\sqrt{2\dot{P}_{0}}\sqrt{r^{2}-2t^{2}\dot
{P}_{0}}\cos\theta)/r
\end{array}
\right.
\end{equation}
in $r>2t\sqrt{\dot{P}_{0}}$, where $\rho_{0}>0$ is an arbitrary parameter,
$\dot{P}_{0}=\dot{P}(\rho_{0})$, $x=r\cos\theta$ and $y=r\sin\theta.$

In this article, we combine the above results to construct solutions with
rotation for the $2$D isothermal Euler-Poisson equations. Our main
contribution is in applying the isothermal pressure term to balance the
potential force term to generate novel solutions.

\begin{theorem}
\label{2-dRotation}For the isothermal ($\gamma=1$) Euler-Poisson equations
(\ref{Euler-Poisson2DRotation}) in $R^{2}$, there exists a family of global
solutions with rotation in radial symmetry,%
\begin{equation}
\left\{
\begin{array}
[c]{c}%
\rho(t,\vec{x})=\rho(t,r)=\frac{1}{a(t)^{2}}e^{f(\frac{r}{a(t)})}\text{,
}{\normalsize u}_{1}{\normalsize =}\frac{\overset{\cdot}{a}(t)}{a(t)}%
x-\frac{\xi}{a(t)^{2}}y\text{, }u_{2}=\frac{\xi}{a(t)^{2}}x+\frac
{\overset{\cdot}{a}(t)}{a(t)}y,\\
\ddot{a}(t)=\frac{-\lambda}{a(t)}+\frac{\xi^{2}}{a(t)^{3}}\text{, }%
a(0)=a_{0}>0\text{, }\dot{a}(0)=a_{1}\\
\overset{\cdot\cdot}{f}(s){\normalsize +}\frac{1}{s}\overset{\cdot}%
{f}(s){\normalsize +\frac{2\pi}{K}e}^{f(s)}{\normalsize =}\frac{2\lambda}%
{K}{\normalsize ,}\text{ }f(0)=\alpha,\text{ }\overset{\cdot}{f}(0)=0\text{,}%
\end{array}
\right.  \label{2-DIsothermalRotation}%
\end{equation}
with arbitrary constants $\xi\neq0,$ $a_{0}$, $a_{1}$ and $\alpha$.

\begin{itemize}
\item[\textrm{(I)}] With $\lambda>0$,\newline\textrm{(a)\,} solutions
(\ref{2-DIsothermalRotation}) are non-trivially time-periodic, except for the
case $a_{0}=\frac{\left\vert \xi\right\vert }{\sqrt{\lambda}}$and $a_{1}%
=0$;\newline\textrm{(b)\,} if $a_{0}=\frac{\left\vert \xi\right\vert }%
{\sqrt{\lambda}}$ and $a_{1}=0$, solutions (\ref{2-DIsothermalRotation}) are steady.

\item[\textrm{(II)}] With $\lambda\leq0$, solutions
(\ref{2-DIsothermalRotation}) are global in time.
\end{itemize}
\end{theorem}

Here, $2$D rotational solutions (\ref{2-DIsothermalRotation}) of the
Euler-Poisson equations (\ref{Euler-Poisson2DRotation}) may be reference
examples for modeling the formation of plate and spiral galaxies or gaseous
stars in the non-relativistic content, because most of the matter is gas at
the early stage of their evolution. Readers can refer to \cite{ZZ} for the
detail description of astrophysical situations. In addition, solutions
(\ref{2-DIsothermalRotation}) may also be applied to the development of
gas-rich and disk-like (dwarf) galaxies \cite{BT}.

\noindent\textbf{Remark.} \emph{By taking $\xi=0$ for solutions
(\ref{2-DIsothermalRotation}) in Theorem \ref{2-dRotation}, we obtain Yuen's
non-rotational solutions (\ref{solution 3 Yuen}), which blow up in a finite
time $T$ if $\lambda>0$. However, the rotational (when $\xi\neq0$) term in
(\ref{2-DIsothermalRotation}) prevents the blowup phenomena.}

\section{Periodic and Spiral Solutions}

Our main work is to design the relevant functions with rotation to fit the 2D
mass equation (\ref{Euler-Poisson2DRotation})$_{1}$.

\begin{lemma}
\label{lem:generalsolutionformasseqrotation2d}For the 2D equation of
conservation of mass
\begin{equation}
\rho_{t}+\nabla\cdot\left(  \rho\vec{u}\right)  =0,
\label{massequationspherical2Drotation}%
\end{equation}
there exist the following solutions:%
\begin{equation}
\rho(t,\vec{x})=\rho(t,r)=\frac{f\left(  \frac{r}{a(t)}\right)  }{a(t)^{2}%
},\text{ }{\normalsize u}_{1}{\normalsize =}\frac{\overset{\cdot}{a}(t)}%
{a(t)}x-\frac{G(t,r)}{r}y\text{, }u_{2}=\frac{G(t,r)}{r}x+\frac{\overset
{\cdot}{a}(t)}{a(t)}y \label{generalsolutionformassequation2Drotation}%
\end{equation}
with arbitrary $C^{1}$ functions $f(s)\geq0$ and $G(t,r)$ and $a(t)>0\in
C^{1}.$
\end{lemma}

\begin{proof}
We plug the following functional form
\begin{equation}
\rho(t,\vec{x})=\rho(t,r)=\frac{f\left(  \frac{r}{a(t)}\right)  }{a(t)^{2}%
},\text{ }{\normalsize u}_{1}{\normalsize =}\frac{F(t,r)}{r}x-\frac{G(t,r)}%
{r}y\text{, }u_{2}=\frac{G(t,r)}{r}x+\frac{F(t,r)}{r}y
\end{equation}
with arbitrary $C^{1}$ functions $f(s)\geq0$, $F(t,r)$, $G(t,r)$ and
$a(t)>0\in C^{1}$, into the 2D mass equation
(\ref{massequationspherical2Drotation}) to have
\begin{equation}
\rho_{t}+\nabla\cdot\left(  \rho\vec{u}\right)
\end{equation}%
\begin{equation}
=\rho_{t}+\frac{\partial}{\partial x}\left(  \rho\frac{Fx}{r}-\rho\frac{Gy}%
{r}\right)  +\frac{\partial}{\partial y}\left(  \rho\frac{Fy}{r}+\rho\frac
{Gx}{r}\right)
\end{equation}%
\begin{align}
&  =\rho_{t}+\left(  \frac{\partial}{\partial x}\rho\right)  \frac{Fx}{r}%
+\rho\left(  \frac{\partial}{\partial x}\frac{Fx}{r}\right)  -\left(
\frac{\partial}{\partial x}\rho\right)  \frac{Gy}{r}-\rho\left(
\frac{\partial}{\partial x}\frac{Gy}{r}\right) \nonumber\\
&  +\left(  \frac{\partial}{\partial y}\rho\right)  \frac{Fy}{r}+\rho\left(
\frac{\partial}{\partial y}\frac{Fy}{r}\right)  +\left(  \frac{\partial
}{\partial y}\rho\right)  \frac{Gx}{r}+\rho\left(  \frac{\partial}{\partial
y}\frac{Gx}{r}\right)
\end{align}%
\begin{align}
&  =\rho_{t}+\rho_{r}\frac{x}{r}\frac{Fx}{r}+\rho\left(  F_{r}\frac{x}%
{r}\right)  \frac{x}{r}+\rho\frac{F}{r}-\rho Fx\frac{x}{r^{3}}\nonumber\\
&  -\rho_{r}\frac{x}{r}\frac{Gy}{r}-\rho G_{r}\frac{x}{r}\frac{y}{r}+\rho
Gy\frac{x}{r^{3}}+\rho_{r}\frac{y}{r}\frac{Fy}{r}+\rho\left(  F_{r}\frac{y}%
{r}\right)  \frac{y}{r}\nonumber\\
&  +\rho\frac{F}{r}-\rho Fy\frac{y}{r^{3}}+\rho_{r}\frac{y}{r}\frac{Gx}%
{r}+\rho\left(  G_{r}\frac{y}{r}\right)  \frac{x}{r}-\rho Gx\frac{y}{r^{3}}%
\end{align}%
\begin{align}
&  =\rho_{t}+\rho_{r}\frac{x}{r}\frac{Fx}{r}+\rho\left(  F_{r}\frac{x}%
{r}\right)  \frac{x}{r}+\rho\frac{F}{r}-\rho Fx\frac{x}{r^{3}}\nonumber\\
&  +\rho_{r}\frac{y}{r}\frac{Fy}{r}+\rho\left(  F_{r}\frac{y}{r}\right)
\frac{y}{r}+\rho\frac{F}{r}-\rho Fy\frac{y}{r^{3}}%
\end{align}%
\begin{equation}
=\rho_{t}+\rho_{r}F+\rho F_{r}+\rho F\frac{1}{r}. \label{tomassradial}%
\end{equation}
Then we take the self-similar structure for the density function%
\begin{equation}
\rho(t,\vec{x})=\rho(t,r)=\frac{f\left(  \frac{r}{a(t)}\right)  }{a(t)^{2}%
}\text{,}%
\end{equation}
and $F(t,r)=\frac{\dot{a}(t)}{a(t)}r$ for the velocity $\vec{u}$ to balance
equation (\ref{tomassradial}):%
\begin{equation}
=\frac{\partial}{\partial t}\frac{f\left(  \frac{r}{a(t)}\right)  }{a(t)^{2}%
}+\left(  \frac{\partial}{\partial r}\frac{f\left(  \frac{r}{a(t)}\right)
}{a(t)^{2}}\right)  \frac{\dot{a}(t)r}{a(t)}+\frac{f\left(  \frac{r}%
{a(t)}\right)  }{a(t)^{2}}\frac{\dot{a}(t)}{a(t)}+\frac{f\left(  \frac
{r}{a(t)}\right)  }{a(t)^{2}}\frac{\dot{a}(t)}{a(t)}%
\end{equation}%
\begin{align}
&  =\frac{-2\overset{\cdot}{a}(t)f\left(  \frac{r}{a(t)}\right)  }{a(t)^{3}%
}-\frac{\overset{\cdot}{a}(t)r\overset{\cdot}{f}\left(  \frac{r}{a(t)}\right)
}{a(t)^{4}}\nonumber\\
&  +\frac{\overset{\cdot}{f}\left(  \frac{r}{a(t)}\right)  }{a(t)^{3}}%
\frac{\overset{\cdot}{a}(t)r}{a(t)}+\frac{f\left(  \frac{r}{a(t)}\right)
}{a(t)^{2}}\frac{\overset{\cdot}{a}(t)}{a(t)}+\frac{f\left(  \frac{r}%
{a(t)}\right)  }{a(t)^{2}}\frac{\overset{\cdot}{a}(t)}{a(t)}%
\end{align}%
\begin{equation}
=0.
\end{equation}
The proof is completed.
\end{proof}

The following Lemma is required to show the cyclic phenomena of the rotational
solutions (\ref{2-DIsothermalRotation}).

\begin{lemma}
\label{lemma22--2drotation}With $\xi\neq0$, for the Emden equation%
\begin{equation}
\ddot{a}(t)=\frac{-\lambda}{a(t)}+\frac{\xi^{2}}{a(t)^{3}},\text{ }%
a(0)=a_{0}>0,\text{ }\dot{a}(0)=a_{1}\text{,} \label{eq124-2drotation}%
\end{equation}

\begin{itemize}
\item[\textrm{(I)}] with $\lambda>0$, the solution is non-trivially periodic,
except for the case with $a_{0}=\frac{\left\vert \xi\right\vert }%
{\sqrt{\lambda}}$ and $a_{1}=0$;

\item[\textrm{(II)}] with $\lambda\leq0$, the solution is global.
\end{itemize}
\end{lemma}

\begin{proof}
The proof is standard and similar to Lemma 3 in \cite{YuenCQG2009} for the
Euler-Poisson equations with a negative cosmological constant.\newline(I) For
equation (\ref{eq124-2drotation}), we could multiply $\dot{a}(t)$ and
integrate it in the following manner:%
\begin{equation}
\frac{\dot{a}(t)^{2}}{2}+\lambda\ln a(t)+\frac{\xi^{2}}{2a(t)^{2}}=\theta
\end{equation}
with the constant $\theta=\frac{a_{1}^{2}}{2}+\lambda\ln a_{0}+\frac{\xi^{2}%
}{2a_{0}^{2}}$.\newline Then, we could define the kinetic energy as
\begin{equation}
F_{kin}=\frac{\dot{a}(t)^{2}}{2}%
\end{equation}
and the potential energy as%
\begin{equation}
F_{pot}=\lambda\ln a(t)+\frac{\xi^{2}}{2a(t)^{2}}.
\end{equation}
Here, the total energy is conserved such that%
\begin{equation}
\frac{d}{dt}(F_{kin}+F_{pot})=0.
\end{equation}
The potential energy function has only one global minimum at $\overset{-}%
{a}=\frac{\left\vert \xi\right\vert }{\sqrt{\lambda}}$ for $a(t)\in
(0,+\infty)$. Therefore, by the classical energy method (in Section 4.3 of
\cite{LS}), the solution for equation (\ref{eq124-2drotation}) has a closed
trajectory. The time for traveling the closed orbit is%
\begin{equation}
T=2\int_{t_{1}}^{t_{2}}dt=2\int_{a_{\min}}^{a_{\max}}\frac{da(t)}%
{\sqrt{2\left[  \theta-\left(  \lambda\ln a(t)+\frac{\xi^{2}}{2a(t)^{2}%
}\right)  \right]  }}\text{,}\label{hk12DRotation}%
\end{equation}
where $a(t_{1})=a_{\min}=\underset{t\geq0}{\inf}(a(t))$ and $a(t_{2})=a_{\max
}=\underset{t\geq0}{\sup}(a(t))$ with some constants $t_{1}$, $t_{2\text{ }}%
$such that $t_{2}\geq t_{1}\geq0$.\newline We let $H(t)=\theta-\left(
\lambda\ln a(t)+\frac{\xi^{2}}{2a(t)^{2}}\right)  $, $H_{0}=\left\vert
\theta-\left(  \lambda\ln(a_{\min}+\epsilon)+\frac{\xi^{2}}{2(a_{\min
}+\epsilon)^{2}}\right)  \right\vert $ , and $H_{1}=\left\vert \theta-\left(
\lambda\ln(a_{\max}-\epsilon)+\frac{\xi^{2}}{2(a_{\max}-\epsilon)^{2}}\right)
\right\vert $. Except for the case with $a_{0}=\frac{\left\vert \xi\right\vert
}{\sqrt{\lambda}}$ and $a_{1}=0$, the time in equation (\ref{hk12DRotation})
can be estimated by%
\begin{align}
T &  =\int_{a_{\min}}^{a_{\min}+\epsilon}\frac{2da(t)}{\sqrt{2\left[
\theta-\left(  \lambda\ln a(t)+\frac{\xi^{2}}{2a(t)^{2}}\right)  \right]  }%
}+\int_{a_{\min}+\epsilon}^{a_{\max}-\epsilon}\frac{2da(t)}{\sqrt{2\left[
\theta-\left(  \lambda\ln a(t)+\frac{\xi^{2}}{2a(t)^{2}}\right)  \right]  }%
}\nonumber\\
&  +\int_{a_{\max}-\epsilon}^{a_{\max}}\frac{2da(t)}{\sqrt{2\left[
\theta-\left(  \lambda\ln a(t)+\frac{\xi^{2}}{2a(t)^{2}}\right)  \right]  }}%
\end{align}
with a sufficient small constant $\epsilon>0$,%
\begin{align}
&  \leq\underset{a_{\min}\leq a(t)\leq a_{\min}+\epsilon}{\sup}\left\vert
\frac{1}{-\frac{\lambda}{a(t)}+\frac{\xi^{2}}{a(t)^{3}}}\right\vert \int
_{0}^{H_{0}}\frac{\sqrt{2}dH(t)}{\sqrt{H(t)}}+\int_{a_{\min}+\epsilon
}^{a_{\max}-\epsilon}\frac{2da(t)}{\sqrt{2\left[  \theta-\left(  \lambda\ln
a(t)+\frac{\xi^{2}}{2a(t)^{2}}\right)  \right]  }}\nonumber\\
&  +\underset{a_{\max}-\epsilon\leq a(t)\leq a_{\max}}{\sup}\left\vert
\frac{1}{-\frac{\lambda}{a(t)}+\frac{\xi^{2}}{a(t)^{3}}}\right\vert \int
_{0}^{H_{1}}\frac{\sqrt{2}dH(t)}{\sqrt{H(t)}}%
\end{align}%
\begin{align}
&  =\underset{a_{\min}\leq a\leq a_{\min}+\epsilon}{\sup}\left\vert \frac
{1}{-\frac{\lambda}{a(t)}+\frac{\xi^{2}}{a(t)^{3}}}\right\vert 2\sqrt{2}%
\sqrt{H_{0}}+\int_{a_{\min}+\epsilon}^{a_{\max}-\epsilon}\frac{2da(t)}%
{\sqrt{2\left[  \theta-\left(  \lambda\ln a(t)+\frac{\xi^{2}}{2a(t)^{2}%
}\right)  \right]  }}\nonumber\\
&  +\underset{a_{\max}-\epsilon\leq a(t)\leq a_{\max}}{\sup}\left\vert
\frac{1}{-\frac{\lambda}{a(t)}+\frac{\xi^{2}}{a(t)^{3}}}\right\vert 2\sqrt
{2}\sqrt{H_{1}}%
\end{align}%
\begin{equation}
<\infty.
\end{equation}
Therefore, we have (a) the solutions to the Emden equation
(\ref{eq124-2drotation}) are non-trivially periodic except for the case with
$a_{0}=\frac{\left\vert \xi\right\vert }{\sqrt{\lambda}}$ and $a_{1}%
=0$.\newline Figure 2 below shows a particular solution for the Emden
equation:%
\begin{equation}
\left\{
\begin{array}
[c]{c}%
\ddot{a}(t)=\frac{-1}{a(t)}+\frac{1}{a(t)^{3}}\\
a(0)=1,\text{ }\dot{a}(0)=1.
\end{array}
\right.  \label{grapha(t)}%
\end{equation}

It is clear to see (b) if $a_{0}=\frac{\left\vert \xi\right\vert }%
{\sqrt{\lambda}}$ and $a_{1}=0$, the solutions (\ref{eq124-2drotation}) are steady.

By applying the similar analysis, we can show that\newline(II) with
$\lambda\leq0$, the solutions are global.\newline The proof is completed.
\end{proof}

After obtaining the above two lemmas, we can construct the periodic and spiral
solutions with rotation to the $2$D isothermal Euler-Poisson system
(\ref{Euler-Poisson2DRotation}) as follows.

\begin{proof}
[Proof of Theorem \ref{2-dRotation}]The procedure of the proof is similar to
the proof for the non-rotational fluids \cite{YuenJMAA2008a}. It is clear that
our functions (\ref{2-DIsothermalRotation}) satisfy Lemma
\ref{lem:generalsolutionformasseqrotation2d} for the mass equation
(\ref{Euler-Poisson2DRotation})$_{1}$. For the first momentum equation
(\ref{Euler-Poisson2DRotation})$_{2,1}$, we get%
\begin{equation}
\rho\left[  \frac{\partial u_{1}}{\partial t}+u_{1}\frac{\partial u_{1}%
}{\partial x}+u_{2}\frac{\partial u_{1}}{\partial y}\right]  +\frac{\partial
}{\partial x}P+\rho\frac{\partial}{\partial x}\Phi
\end{equation}%
\begin{equation}
=\rho\left[  \frac{\partial u_{1}}{\partial t}+u_{1}\frac{\partial u_{1}%
}{\partial x}+u_{2}\frac{\partial u_{1}}{\partial y}\right]  +\frac{\partial
}{\partial x}\frac{Ke^{f(\frac{r}{a(t)})}}{a(t)^{2}}+\rho\frac{\partial
}{\partial x}\Phi\text{.}%
\end{equation}
By defining the variable $s=\frac{r}{a(t)}$ with $\frac{\partial}{\partial
x}=\frac{\partial}{\partial r}\frac{\partial r}{\partial x}=\frac{x}{r}%
\frac{\partial}{\partial r}$, we have%
\begin{equation}
=\rho\left[  \frac{\partial u_{1}}{\partial t}+u_{1}\frac{\partial u_{1}%
}{\partial x}+u_{2}\frac{\partial u_{1}}{\partial y}\right]  +\frac{Ke^{f(s)}%
}{a(t)^{2}}\frac{x}{r}\frac{\partial}{a(t)\partial\left(  \frac{r}%
{a(t)}\right)  }f(s)+\frac{x}{r}\rho\frac{\partial}{\partial r}\Phi
\end{equation}%
\begin{equation}
=\rho\left[  \frac{\partial u_{1}}{\partial t}+u_{1}\frac{\partial u_{1}%
}{\partial x}+u_{2}\frac{\partial u_{1}}{\partial y}+\frac{K}{a(t)}\frac{x}%
{r}\frac{\partial}{\partial s}f(s)+\frac{x}{r}\frac{2\pi}{r}%
%TCIMACRO{\dint \limits_{0}^{r}}%
%BeginExpansion
{\displaystyle\int\limits_{0}^{r}}
%EndExpansion
\frac{e^{f(\frac{\eta}{a(t)})}}{a(t)^{2}}\eta d\eta\right]
\end{equation}

\begin{equation}
=\rho\left[
\begin{array}
[c]{c}%
\frac{\partial}{\partial t}\left(  \frac{\dot{a}(t)}{a(t)}x-\frac{\xi
}{a(t)^{2}}y\right)  +\left(  \frac{\dot{a}(t)}{a(t)}x-\frac{\xi}{a(t)^{2}%
}y\right)  \frac{\partial}{\partial x}\left(  \frac{\dot{a}(t)}{a(t)}%
x-\frac{\xi}{a(t)^{2}}y\right)  \\
+\left(  \frac{\xi}{a(t)^{2}}x+\frac{\dot{a}(t)}{a(t)}y\right)  \frac
{\partial}{\partial y}\left(  \frac{\dot{a}(t)}{a(t)}x-\frac{\xi}{a(t)^{2}%
}y\right)  +\frac{x}{a(t)r}\left(  K\dot{f}(s)+\frac{2\pi}{\frac{r}{a(t)}}%
%TCIMACRO{\dint \limits_{0}^{r}}%
%BeginExpansion
{\displaystyle\int\limits_{0}^{r}}
%EndExpansion
\frac{e^{f(\frac{\eta}{a(t)})}}{a(t)^{2}}\eta d\eta\right)
\end{array}
\right]
\end{equation}%
\begin{equation}
=\rho\left[
\begin{array}
[c]{c}%
\left(  \frac{\ddot{a}(t)}{a(t)}-\frac{\dot{a}(t)^{2}}{a(t)^{2}}\right)
x+\frac{2\xi\dot{a}(t)}{a(t)^{3}}y+\left(  \frac{\dot{a}(t)}{a(t)}x-\frac{\xi
}{a(t)^{2}}y\right)  \frac{\dot{a}(t)}{a(t)}\\
-\left(  \frac{\xi}{a(t)^{2}}x+\frac{\dot{a}(t)}{a(t)}y\right)  \frac{\xi
}{a(t)^{2}}+\frac{x}{a(t)r}\left(  K\dot{f}(s)+\frac{2\pi}{\frac{r}{a(t)}}%
%TCIMACRO{\dint \limits_{0}^{r}}%
%BeginExpansion
{\displaystyle\int\limits_{0}^{r}}
%EndExpansion
e^{f(\frac{\eta}{a(t)})}\left(  \frac{\eta}{a(t)}\right)  d\left(  \frac{\eta
}{a(t)}\right)  \right)
\end{array}
\right]
\end{equation}%
\begin{equation}
=\frac{x\rho}{a(t)r}\left[  \left(  \ddot{a}(t)-\frac{\xi^{2}}{a(t)^{3}%
}\right)  r+K\dot{f}(s)+\frac{2\pi}{s}%
%TCIMACRO{\dint \limits_{0}^{s}}%
%BeginExpansion
{\displaystyle\int\limits_{0}^{s}}
%EndExpansion
e^{f(\tau)}\tau d\tau\right]
\end{equation}%
\begin{equation}
=\frac{x\rho}{a(t)r}\left[  -\lambda s+K\dot{f}(s)+\frac{2\pi}{s}%
%TCIMACRO{\dint \limits_{0}^{s}}%
%BeginExpansion
{\displaystyle\int\limits_{0}^{s}}
%EndExpansion
e^{f(\tau)}\tau d\tau\right]  \label{eq581}%
\end{equation}
with the Emden equation
\begin{equation}
\left\{
\begin{array}
[c]{c}%
\ddot{a}(t)=\frac{-\lambda}{a(t)}+\frac{\xi^{2}}{a(t)^{3}}\\
a(0)=a_{0}>0\text{, }\dot{a}(0)=a_{1}\text{,}%
\end{array}
\right.  \label{Endemeqeq2-Drotation}%
\end{equation}
with an arbitrary constant $\xi\neq0$.\newline Similarly, we obtain the
corresponding result for the second momentum equation
(\ref{Euler-Poisson2DRotation})$_{2,2}$ in the following manner with
$\frac{\partial}{\partial y}=\frac{\partial}{\partial r}\frac{\partial
r}{\partial y}=\frac{y}{r}\frac{\partial}{\partial r}$:%
\begin{equation}
\rho\left[  \frac{\partial u_{2}}{\partial t}+u_{1}\frac{\partial u_{2}%
}{\partial x}+u_{2}\frac{\partial u_{2}}{\partial y}\right]  +\frac{\partial
}{\partial y}P+\rho\frac{\partial}{\partial y}\Phi
\end{equation}%
\begin{equation}
=\rho\left[  \frac{\partial u_{2}}{\partial t}+u_{1}\frac{\partial u_{2}%
}{\partial x}+u_{2}\frac{\partial u_{2}}{\partial y}\right]  +\frac{Ke^{f(s)}%
}{a(t)^{2}}\frac{y}{r}\frac{\partial}{a(t)\partial\left(  \frac{r}%
{a(t)}\right)  }f(s)+\frac{y}{r}\rho\frac{\partial}{\partial r}\Phi
\end{equation}%
\begin{equation}
=\rho\left[  \frac{\partial u_{2}}{\partial t}+u_{1}\frac{\partial u_{2}%
}{\partial x}+u_{2}\frac{\partial u_{2}}{\partial y}+\frac{K}{a(t)}\frac{y}%
{r}\frac{\partial}{\partial s}f(s)+\frac{y}{r}\frac{2\pi}{r}%
%TCIMACRO{\dint \limits_{0}^{r}}%
%BeginExpansion
{\displaystyle\int\limits_{0}^{r}}
%EndExpansion
\frac{e^{f(\frac{\eta}{a(t)})}}{a(t)^{2}}\eta d\eta\right]
\end{equation}%
\begin{equation}
=\rho\left[
\begin{array}
[c]{c}%
\frac{\partial}{\partial t}\left(  \frac{\xi}{a(t)^{2}}x+\frac{\overset{\cdot
}{a}(t)}{a(t)}y\right)  +\left(  \frac{\dot{a}(t)}{a(t)}x-\frac{\xi}{a(t)^{2}%
}y\right)  \frac{\partial}{\partial x}\left(  \frac{\xi}{a(t)^{2}}%
x+\frac{\overset{\cdot}{a}(t)}{a(t)}y\right)  \\
+\left(  \frac{\xi}{a(t)^{2}}x+\frac{\dot{a}(t)}{a(t)}y\right)  \frac
{\partial}{\partial y}\left(  \frac{\xi}{a(t)^{2}}x+\frac{\overset{\cdot}%
{a}(t)}{a(t)}y\right)  +\frac{y}{a(t)r}\left(  K\dot{f}(s)+\frac{2\pi}%
{\frac{r}{a(t)}}%
%TCIMACRO{\dint \limits_{0}^{r}}%
%BeginExpansion
{\displaystyle\int\limits_{0}^{r}}
%EndExpansion
\frac{e^{f(\frac{\eta}{a(t)})}}{a(t)^{2}}\eta d\eta\right)
\end{array}
\right]
\end{equation}%
\begin{equation}
=\rho\left[
\begin{array}
[c]{c}%
-\frac{2\xi\dot{a}(t)}{a(t)^{3}}x+\left(  \frac{\ddot{a}(t)}{a(t)}-\frac
{\dot{a}(t)^{2}}{a(t)^{2}}\right)  y+\left(  \frac{\dot{a}(t)}{a(t)}%
x-\frac{\xi}{a(t)^{2}}y\right)  \frac{\xi}{a(t)^{2}}\\
+\left(  \frac{\xi}{a(t)^{2}}x+\frac{\dot{a}(t)}{a(t)}y\right)  \frac
{\overset{\cdot}{a}(t)}{a(t)}+\frac{y}{a(t)r}\left(  K\dot{f}(s)+\frac{2\pi
}{\frac{r}{a(t)}}%
%TCIMACRO{\dint \limits_{0}^{r}}%
%BeginExpansion
{\displaystyle\int\limits_{0}^{r}}
%EndExpansion
e^{f(\frac{\eta}{a(t)})}\left(  \frac{\eta}{a(t)}\right)  d\left(  \frac{\eta
}{a(t)}\right)  \right)
\end{array}
\right]
\end{equation}%
\begin{equation}
=\frac{y\rho}{a(t)r}\left[  \left(  \ddot{a}(t)-\frac{\xi^{2}}{a(t)^{3}%
}\right)  r+K\dot{f}(s)+\frac{2\pi}{s}%
%TCIMACRO{\dint \limits_{0}^{s}}%
%BeginExpansion
{\displaystyle\int\limits_{0}^{s}}
%EndExpansion
e^{f(\tau)}\tau d\tau\right]
\end{equation}%
\begin{equation}
=\frac{y\rho}{a(t)r}\left[  -\lambda s+K\dot{f}(s)+\frac{2\pi}{s}%
%TCIMACRO{\dint \limits_{0}^{s}}%
%BeginExpansion
{\displaystyle\int\limits_{0}^{s}}
%EndExpansion
e^{f(\tau)}\tau d\tau\right]  \text{.}\label{eq58}%
\end{equation}
To make equations (\ref{eq581}) and (\ref{eq58}) equal zero, we may require
the Liouville equation from differential geometry:%
\begin{equation}
\left\{
\begin{array}
[c]{c}%
\ddot{f}(s)+\frac{\dot{f}(s)}{s}+\frac{2\pi}{K}e^{f(s)}=\frac{2\lambda}{K}\\
f(0)=\alpha\text{, }\dot{f}(0)=0.
\end{array}
\right.  \label{Liouville2-DRotation}%
\end{equation}

We note that the global existence of the initial value problem of the
Liouville equation (\ref{Liouville2-DRotation}) has been shown by Lemma 10 in
\cite{YuenJMAA2008a}. Thus, we confirm that functions
(\ref{2-DIsothermalRotation}) are a family of classical solutions for the
isothermal ($\gamma=1$) Euler-Poisson equations (\ref{Euler-Poisson2DRotation}%
) in $R^{2}$.

With Lemma \ref{lemma22--2drotation}, it is clear that\newline(I) With
$\lambda>0$,\newline(a) solutions (\ref{2-DIsothermalRotation}) are
non-trivially time-periodic, except for the case $a_{0}=\frac{\left\vert
\xi\right\vert }{\sqrt{\lambda}}$and $a_{1}=0$;\newline(b) if $a_{0}%
=\frac{\left\vert \xi\right\vert }{\sqrt{\lambda}}$ and $a_{1}=0$, solutions
(\ref{2-DIsothermalRotation}) are steady.\newline(II) With $\lambda\leq0$,
solutions (\ref{2-DIsothermalRotation}) are global in time.\newline Therefore
all of the rotational solutions (\ref{2-DIsothermalRotation}) with $\xi\neq0$,
are global in time.\newline We complete the proof.
\end{proof}

\section{Conclusion and Discussion}

Our results confirm that there exists a class of periodic solutions which can
be found by choosing a sufficiently small constant $a_{0}<<1$ in solutions
(\ref{2-DIsothermalRotation}), in the Euler-Poisson equations
(\ref{Euler-Poisson2DRotation}) in $R^{2}$, even without a negative
cosmological constant \cite{YuenCQG2009}. Here, the periodic rotation prevents
the blowup phenomena that occur in solutions without rotation
\cite{YuenJMAA2008a}.

It is open to show the existences of solutions and their stabilities for the
small perturbation of these solutions (\ref{2-DIsothermalRotation}). Numerical
simulation and mathematical proofs for the perturbational solutions are
suggested for understanding their evolution.

As our solutions in this paper works for the 2D case, the corresponding
rotational solutions in $R^{3}$ are conjectured. We conjecture that the
corresponding rotational solutions to Goldreich and Weber's solutions
(\ref{3-dgamma=4over3}) for the 3D Euler-Poisson equations with $\gamma=4/3$
\cite{GW} exist, such as the ones for the Euler equations
\cite{Yuen3DexactEuler}. Further research is expected to shed more light on
the possibilities.

\section{Acknowledgement}

The authors thank the reviewers and the editors for their helpful comments for
improving the quality of this article.

\end{document}